\pdfoutput=1
\documentclass[a4paper,UKenglish]{lipics-v2019}

\usepackage{microtype}

\title{Logics for Reversible Regular Languages and Semigroups with Involution}

\author{Paul Gastin}{LSV, ENS Paris-Saclay \& CNRS, Universit{\'e} 
Paris-Saclay, France}{paul.gastin@lsv.fr}{}{}

\author{Amaldev Manuel}{Indian Institute of Technology Goa, India}{amal@iitgoa.ac.in}{}{}

\author{R.\ Govind}{Chennai Mathematical Institute, Chennai, India and LaBRI, University of Bordeaux, France}{govindr@cmi.ac.in}{}{}

\authorrunning{P.\ Gastin, M.\ Amaldev, R.\ Govind} 

\Copyright{Paul Gastin and Amaldev Manuel and R.\ Govind}



\funding{Partly supported by UMI ReLaX.}

\nolinenumbers 
\hideLIPIcs  

\usepackage{latexsym}
\usepackage{mathtools}
\usepackage{amssymb}
\usepackage{xspace}

\renewcommand \( {\left(}
\renewcommand \) {\right)}

\newcommand \resp {\mathrm{resp.}}
\newcommand \Rev {\ensuremath{\mathsf{Rev}}\xspace}
\newcommand \mso {\ensuremath{\mathsf{MSO}}\xspace}
\newcommand \fo {\ensuremath{\mathsf{FO}}\xspace}

\newcommand \ltt {\ensuremath{\mathsf{LTT}}\xspace}
\newcommand \lrtt {\ensuremath{\mathsf{LRTT}}\xspace}

\newcommand \bet {\ensuremath \mathsf{bet}}
\newcommand \N {\ensuremath \mathsf{N}}

\newcommand{\approxr}[2]{\stackrel{r}{\approx}\mathrel{{}^{#1}_{#2}}}

\begin{document}
\maketitle
\begin{abstract}
  We present \mso and \fo logics with predicates `between' and `neighbour' that
  characterise various fragments of the class of regular languages that are
  closed under the reverse operation.  The standard connections that exist
  between \mso and \fo logics and varieties of finite semigroups extend to this
  setting with semigroups extended with an involution.  The case is different
  for \fo with neighbour relation where we show that one needs additional
  equations to characterise the class.
\end{abstract}
 
\section{Introduction}

In this paper we look closely at the class of regular languages that are closed
under the reverse operation.  We fix a finite alphabet $A$ for the rest of our
discussion.  The set $A^*$ (respectively $A^+$) denotes the set of all ($\resp$
non-empty) finite words over the alphabet $A$.  If $w = a_1 \cdots a_k$ with
$a_i\in A$ is a word then $w^r = a_k\cdots a_1$ denotes the reverse of $w$.
This notion is extended to sets of words pointwise, i.e.\ $L^r = \{ w^r \mid w
\in L\}$ and we can talk about reverse of languages.  A regular language $L
\subseteq A^*$ is {\em closed under reverse} or simply {\em reversible} if $L^r
= L$.  We let $\Rev$ denote the class of all reversible regular languages.
Clearly $\Rev$ is a strict subset of the class of all regular languages.
 
The class $\Rev$ is easily verified to be closed under union, intersection and
complementation.  It is also closed under homomorphic images, and inverse
homomorphic images under alphabetic (i.e.\ length preserving) morphisms.  However
they are not closed under quotients.  For instance, the language
$L=(abc)^*+(cba)^*$ is closed under reverse but the quotient $a^{-1}L =
bc(abc)^*$ is not closed under reverse.  Thus the class $\Rev$ fails to be a
{\em variety} of languages --- i.e.\ a class closed under Boolean operations,
inverse morphic images and quotients.  However reversible languages are closed
under bidirectional quotients, i.e.\ quotients of the form $u^{-1}L v^{-1} \cup
\(v^r\)^{-1}L \(u^r\)^{-1}$, given words $u,v$.  Thus, to a good extent, $\Rev$
shares properties similar to that of regular languages.  Hence it makes sense to
ask the question
 
\begin{quote}
  {\em ``are there good logical characterisations for the class $\Rev$ and its
  well behaved subclasses?''}.
\end{quote}

\subparagraph*{Our results.} We suggest a positive answer to the above question.  We
introduce two predicates {\em between} ($\bet(x,y,z)$ is true if position $y$ is
between positions $x$ and $z$) and {\em neighbour} ($\N(x,y)$ is true if
positions $x$ and $y$ are adjacent).  The predicates {\em between} and {\em
neighbour} are the natural analogues of the order relation $<$ and successor
relation $+1$ in the undirected case.  In fact this analogy extends to the case
of logical definability.  We show that $\Rev$ is the class of monadic second
order (\mso) definable languages using either of the predicates, i.e.\
$\mso(\bet)$ or $\mso(\N)$.  This is analogous to the classical
B\"uchi-Elgot-Trakhtenbrot theorem relating regular languages and \mso logic.
This connection extends to the case of first order logic as well.  We show that
$\fo(\bet)$ definable languages are precisely the reversible languages definable
in $\fo(<)$.  However the case of successor relation is different, i.e.\ the
class of $\fo(\N)$ definable languages is a strict subset of reversible
languages definable in $\fo(+1)$.  The precise characterisation of this class is
one of our main contributions.

The immediate question that arises from the above characterisations is one of
definability: {\em Given a reversible language is it definable in the logic?"}.
The case of $\fo(\bet)$ is decidable due to Sch\"utzenberger-McNaughton-Papert
theorem that states that syntactic monoids of $\fo(<)$ definable languages are
aperiodic (equivalent to the condition that the monoid contains no groups as
subsemigroups) \cite{Schutzenberger,McNaughton}.  However the question for
$\fo(\N)$ is open.  We prove a partial characterisation in terms of semigroups
with involution.  It is to be noted that the characterisation of $\fo(+1)$ is a
tedious one that goes via categories \cite{StraubingBook}.

\subparagraph*{Related work.} A different but related {\em between} predicate
(namely $a(x,y)$, for $a \in A$, is true if there is an $a$-labelled position
between positions $x$ and $y$) was introduced and studied in
\cite{Straubing1,Straubing2,Straubing3}.  Such a predicate is not definable in
$\fo^2(<)$, the two variable fragment of first-order logic (which corresponds to
the well known semigroup variety DA \cite{Tesson}).  The authors of
\cite{Straubing1,Straubing2,Straubing3} study the expressive power of $\fo^2(<)$
enriched with the between predicates $a(x,y)$ for $a\in A$, and show an
algebraic characterisation of the resulting family of languages.  The between
predicate (predicates rather) in \cite{Straubing1} is strictly less expressive
than the between predicate introduced in this paper.  However the logics
considered in \cite{Straubing1} have the between predicates in conjunction with
order predicates $<$ and $+1$.  Hence their results are orthogonal to ours.
 
Another line of work that has close parallels with the one in this paper is the
variety theory of involution semigroups (also called $\star$-semigroups) (see
\cite{Dolinka} for a survey).  Most investigations along these lines have been
on subvarieties of {\em regular} $\star$-semigroups (i.e.\ $\star$-semigroups
satisfying the equation $xx^\star x=x$).  As far as we are aware the equation
introduced in this paper has not been studied before.
 
\subparagraph*{Structure of the paper.} In Section \ref{Section:logic} we introduce
the predicates and present our logical characterisations.  This is followed by a
characterisation of $\fo(\N)$.  In Section \ref{Section:semigroup} we discuss
semigroups with involution, a natural notion of syntactic semigroups for
reversible languages.  In Section \ref{Section:conclusion} we conclude.

\section{Logics with {\em Between} and {\em Neighbour}}
\label{Section:logic}

As usual we represent a word $w = a_1 \cdots a_n$ as a structure containing
positions $\{1, \ldots, n\}$, and unary predicates $P_a$ for each letter $a$ in
the alphabet.  The predicate $P_a$ is precisely true at those positions labelled
by letter $a$.  The atomic predicate $x<y$ ($\resp$ $x+1=y$) is true if position
$y$ is after ($\resp$ immediately after) position $x$.  The logic \fo is the
logic containing atomic predicates, boolean combinations ($\phi \vee \psi$,
$\phi \wedge \psi$, $\neg \psi$ whenever $\phi, \psi$ are formulas of the
logic), and first order quantifications ($\exists x\,\psi$, $\forall x \,\psi$
if $\psi$ is a formula of the logic).  The logic \mso in addition contains
second order quantification as well ($\exists X\,\psi$, $\forall X \,\psi$ if
$\psi$ is a formula of the logic) --- i.e.\ quantification over sets of
positions.  By $\fo(\tau)$ or $\mso(\tau)$ we mean the corresponding logic with
atomic predicates $\tau$ in addition to the unary predicates $P_a$.  The
classical result relating \mso and regular languages states that $\mso(<) =
\mso(+1)$ defines all regular languages.  We introduce two analogous predicates
for the class $\Rev$ of reversible regular languages.
 
\subsection{$\mso(\bet), \mso(\N)$ and $\fo(\bet)$}

The ternary {\em between} predicate $\bet(x,y,z)$ is true for positions $x, y,
z$ when $y$ is in between $x$ and $z$, i.e.\
$$\bet(x,y,z) ~~:=~~ x < y < z \text{ or } z < y < x.$$

\begin{example} \label{Example:sigma1}
  The set of all words containing the subword $a_1a_2 \cdots a_k$ or $a_ka_{k-1}
  \cdots a_1$ is defined by the formula
  $$\exists x_1 \exists x_2 \cdots \exists x_k\ \bigwedge\limits_{i=1}^k  P_{a_i}(x_i)\ 
  \wedge \ \bigwedge\limits_{i=2}^{k-1}\ bet(x_{i-1},x_i,x_{i+1}).$$
\end{example}

The `successor' relation of $\bet$ is the binary predicate {\em neighbour}
$\N(x,y)$ that holds true when $x$ and $y$ are neighbours, i.e.
$$\N(x,y) ~~:=~~ x +1 =  y \text{ or } y+1  = x.$$

\begin{example}
The set of words of even length is defined by the formula

$$\exists X (X(e_1)  \wedge  \neg X(e_2)  
\wedge  \forall x\forall y (\N(x,y) \rightarrow (X(x) \leftrightarrow \neg X(y))))$$
where $e_1, e_2$ are the endpoints, i.e.\ the two positions with exactly one
neighbour (defined easily in $\fo(\N)$).
\end{example}

The relation $\N(x,y)$ can be defined in terms of $\bet$ using first-order
quantifiers as $x\neq y \wedge \forall z\,\neg\bet(x,z,y)$.  One can
also define $\bet(x,y,z)$ in terms of $\N$, but using second-order set
quantification.  To do this we assert that any subset $X$ of positions
\begin{itemize}
  \item that contains $x$, $z$ and at least some other position
  \item and such that any position in $X$, except for $x$ and $z$, has exactly two neighbours in $X$,
\end{itemize}
contains the position $y$.
\begin{proposition}
  For definable languages, $\mso(\bet) = \mso(\N)= \Rev$. 
\end{proposition}
\begin{proof}
  Clearly from the discussion above, $\mso(\bet) = \mso(\N) \subseteq \Rev$.  To
  show the other inclusion, let $L$ be a reversible regular language and let
  $\varphi$ be a formula in $\mso(<)$ defining it.  Pick an endpoint $e$ of the
  given word, an endpoint is a position with exactly one neighbour, a property
  expressible in $\fo(\N)\subseteq\fo(\bet)$.  We relativize the formula
  $\varphi$ with respect to $e$ by replacing all occurrences of $x<y$ in the
  formula by $(e=x\neq y)\vee\bet(e,x,y)$.  Let $\varphi'(e)$ be the formula
  obtained in this way and let
  $\psi(e)=\neg\exists x,y\, (x\neq y \wedge \N(e,x) \wedge \N(e,y))$
  be the $\fo(\N)$ formula asserting that $e$ is an endpoint, then we claim that
  $$
  \chi = \exists e \(\psi\(e\) \wedge \varphi'\(e\)\)
  $$
  defines the language $L$. Let $w$ be a word of length $k\geq1$ then,
  \begin{align*}
    w \models \chi &~\Leftrightarrow~ w, 1 \models \varphi'(e) \text{ or } w, k \models \varphi'(e) \\
    &~\Leftrightarrow~ w \models \varphi \text{ or } w ^r \models \varphi \\
    &~\Leftrightarrow~ w \models \varphi ~~(\text{since $L$ is reversible})
  \end{align*}
  Hence $L(\chi)=L(\varphi)=L$. 
\end{proof}
The above proposition says that $\mso(\bet) = \mso(<) \cap \Rev$.  This carries
down to the first-order case using the same relativization idea.  In fact the
result holds for the prefix class $\Sigma_i$ (first-order formulas in prenex
normal form with $i$ blocks of alternating quantifiers starting with
$\exists$-block).
\begin{proposition}   \label{Prop:fobet}
  The following is true for definable languages.
  \begin{enumerate}
    \item  $\fo(\bet) = \fo(<) \cap  \Rev$. 
    \item $\Sigma_i(\bet)=\Sigma_i(<) \cap \Rev$.
  \end{enumerate}
\end{proposition}
\begin{proof}
  Given an $\fo(<)$ formula in prenex form defining a language in $\Rev$, we
  replace every occurrence of $x<y$ by $(e=x\neq y)\vee\bet(e,x,y)$ as before,
  where $e$ is asserted to be an endpoint with $\psi(e)=\forall
  x,y\,\neg\bet(x,e,y)$.  For every formula in $\Sigma_i(<)$, $i\geq 2$ this
  results in an equivalent formula in $\Sigma_i(\bet)$.  For the case of
  $\Sigma_1$, let us note that every formula in $\Sigma_1(<)$ defines a union of
  languages of the form $A^*a_1A^*a_2 A^*\cdots A^*a_kA^*$.  Such a language can
  be written as a disjunction of formulas like the one in
  Example~\ref{Example:sigma1}.  
\end{proof}

 \subsection{$\fo(\N)$}
Next we address the expressive power of \fo with the neighbour predicate. 

We start by detailing the class of {\em locally threshold testable languages}.
Recall that word $y$ is a factor of word $u$ if $u = x y z$ for some $x,z$ in
$A^*$.  We use $\sharp(u, y)$ to denote the number of times the factor $y$
appears in $u$.

Let $\approx_k^t$, for $k, t>0$, be the equivalence on $A^*$, whereby two words
$u$ and $v$ are equivalent if either they both have length at most $k-1$ and
$u=v$, or otherwise they have
\begin{enumerate}
\item the same prefix of length $k-1$,
\item the same suffix of length $k-1$,
\item and the same number of occurrences, upto threshold $t$, for all factors of length $ \leq k$, i.e.\ 
for each word $y \in A^*$ of length at most  $k$, either 
$\sharp(u, y)  = \sharp(v, y) < t$, or
 $\sharp(u, y) \ge t$ and $\sharp(v,y)\ge t$.
\end{enumerate}

\begin{example}
  We have $ababab \approx_2^1 abab \not\approx_2^1 abbab$.  Indeed, all the
  words start and end with the same letter.  In the first two words the factors
  $ab$ as well as $ba$ appear at least once.  While in the last word the factor
  $bb$ appears once while it is not present in the word $abab$.
  Notice also that $ababab \not\approx_2^2 abab$ due to the factor $ba$.
\end{example}

A language is \emph{locally threshold testable} (or \ltt for short) if it is a
union of $\approx_{k}^{t}$ classes, for some $k,t > 0$.
\begin{example}
  The language $(ab)^*$ is \ltt. In fact it is {\em locally testable} (the
  special case of locally threshold testable with $t=1$).  Indeed, $(ab)^{*}$ is
  the union of three classes: $\{\varepsilon\}$, $\{ab\}$ and $abab(ab)^{*}$
  which is precisely the set of words that begin with $a$, end with $b$, and the
  only factors are $ab$ and $ba$.
  
  A language that is definable in $\fo(<)$ and not \ltt is $c^*ac^*bc^*$.  In
  this language if $a$ and $b$ are sufficiently separated by $c$-blocks then
  the order between $a$ and $b$ cannot be differentiated.  It can be proved that
  for any $t,k$ there is a sufficiently large $n$ such that $c^n a c^n b c^n
  \approx_k^t c^n bc^nac^n$.
\end{example}

Locally threshold testable languages are precisely the class of languages
definable in $\fo(+1)$ \cite{BeauquierPin,WolfgangThomas}.  Since we can define the
neighbour predicate $\N$ using $+1$, clearly $\fo(\N) \subseteq \fo(+1) \cap
\Rev = \ltt \cap \Rev$.  But this inclusion is strict as shown in
Example \ref{fo(n)-containment}.

\begin{example}
	Consider the language $L=ua^* + a^*u^r$ of words which have either $u$ as
	prefix and followed by an arbitrary number of $a$'s, or $u^r$ as suffix and preceded
	by an arbitrary number of $a$'s.  The language $L$ is in $\fo(\N)$. When 
  $u=a_1\cdots a_n$, it can be defined by a formula of the form $\exists 
  x_1,\ldots,x_n\,\psi$ where $\psi$ states that $x_1$ is an endpoint, 
  $\bigwedge_{1\leq i<n}\N(x_i,x_{i+1})$, $\bigwedge_{1<i<n}x_{i-1}\neq 
  x_{i+1}$, $\bigwedge_{1\leq i\leq n} P_{a_i}(x_i)$, and all other positions 
  are labelled $a$.
\end{example}

\begin{example}\label{fo(n)-containment}
  Consider the language $L$ over the alphabet $\{a,b,c\}$,
  $$L = \{ w \mid \sharp(w, ab) = 2 , \sharp(w, ba) = 1 \text{ or }  \sharp(w, ab) = 1, \sharp(w, ba) = 2 \}.$$
  Since $L$ is locally threshold testable and reverse closed, $L \in \fo(+1) \cap \Rev$.

We can show that $L \not\in \fo(\N)$ by showing that the words, 
\begin{center}
$c^k\ ab\ c^k\ ba\ c^kab\ c^k \in L$ \hfil $c^k\ ab\ c^k\ ab\ c^k\ ab\ c^k \not \in L$
\end{center}
for $k>0$ are indistinguishable by an $\fo(\N)$ formula of quantifier depth $k$.
For showing the latter claim, one uses Ehrenfeucht-Fraissé games and argues that
in the $k$-round EF-game the duplicator has a winning strategy.  The strategy is
roughly described below:
\begin{center}
$\underline{c^kabc^kb}\,{ac^kabc^k}$ \hfil ${c^kabc^ka}\,\underline{bc^kabc^k}$
\end{center}
Any move of the spoiler is mimicked by the duplicator in the corresponding
underlined or non-underlined part of the other word, while maintaining the
neighbourhood relation between positions.  For instance, if the spoiler plays
the first $b$ on the underlined part of the first word, then the duplicator
chooses the last $b$ on the underlined portion of the word on the right.
Similarly, if the spoiler plays the first $a$ on the non-underlined part of the
first word, the duplicator chooses the last $a$ on the non-underlined portion of
the word on the right.  Note that, since no order on positions in the words can
be checked with the neighbour predicate, there is no way to distinguish between
these words, if the duplicator plays in the above way ensuring that the position
played has the same neighbourhood relation as the position played by the
spoiler.  Therefore, the Neighbour predicate will not be able to distinguish
between $ab$ and $ba$ when they are sufficiently separated by $c$'s.

\end{example}
From the above example, we get,
\begin{proposition}
\label{Prop:fon}
For definable languages, $\fo(\N) \subsetneq \fo(+1) \,\cap\, \Rev = \ltt \,\cap\, \Rev$.
\end{proposition}

Next we will characterise the class of languages accepted by $\fo(\N)$. 
For $t>0$ we define the equality with threshold $t$ on the set $\mathbb{N}$ of 
natural numbers by $i=^{t}j$ if $i=j$ or $i,j\geq t$. Recall that $\sharp(w,v)$ 
denotes the number of occurrences of $v$ in $w$, i.e.\ the number of pairs 
$(x,y)$ such that $w=xvy$. We extend this to $\sharp^{r}(w,v)$ which counts the 
number of occurrences of $v$ or $v^{r}$ in $w$, i.e.\ the number of pairs 
$(x,y)$ such that $w=xvy$ or $w=xv^{r}y$. Notice that 
$\sharp^{r}(w,v)=\sharp^{r}(w,v^{r})=\sharp^{r}(w^{r},v)=\sharp^{r}(w^{r},v^{r})$.

We define now the \emph{locally-reversible threshold testable} (\lrtt) equivalence relation.  
Let $k,t>0$.  Two words $w,w'\in A^{*}$ are $(k,t)$-\lrtt equivalent, denoted
$w\approxr{t}{k}w'$ if $|w|<k$ and $w'\in\{w,w^{r}\}$, or
\begin{itemize}
  \item $w,w'$ are both of length at least $k$, and
  
  \item $\sharp^{r}(w,v) =^{t} \sharp^{r}(w',v)$ for all $v\in A^{\leq k}$, and

  \item if $x,x'$ are the prefixes of $w,w'$ of length $k-1$ and $y,y'$ are the
  suffixes of $w,w'$ of length $k-1$ then $\{x,y^{r}\}=\{x',y'^{r}\}$.
\end{itemize}
Notice that $w\approxr{t}{k}w^{r}$ for all $w\in A^{*}$ and $w\approx^{t}_k w'$
implies $w\approxr t k w'$ for all $w,w'\in A^{*}$.  Notice also that $\approxr
t k$ is not a congruence.  Indeed, we have $ab \approxr t k ba$ but $aba
\not\approxr t k baa$.  On the other hand, if $v\approxr t k w$ then for all
$u\in A^{*}$ we have $uv \approxr t k uw$ or $uv \approxr t k uw^{r}$, and
similarly $vu \approxr t k wu$ or $vu \approxr t k w^{r}u$.

\begin{definition}[Locally-Reversible Threshold Testable Languages]
  A language $L$ is {\em locally-reversible threshold testable}, \lrtt for
  short, if it is a union of equivalence classes of $\approxr{t}{k}$ for some
  $k,t>0$.
\end{definition}

\begin{theorem}
  Languages defined by $\fo(\N)$ are precisely the class of locally-reversible
  threshold testable languages.
\label{Thm:fon}
\end{theorem}

\begin{proof}
  ($\Leftarrow$) Assume we are given an \lrtt language, i.e.\ a union of 
  $\approxr t k$-classes for some $k,t>0$.  We explain how to write an $\fo(\N)$ formula for
  each $\approxr t k$-class. 
  Consider a word $v=a_1a_2\cdots a_n \in A^{+}$.  For $m\in\mathbb{N}$, we can 
  say that $v$ or its reverse occurs at least $m$ times in a word $w\in A^{*}$,
  i.e.\ $\sharp^{r}(w,v)\geq m$, by the formula 
  \begin{equation*}
    \begin{split}
      \varphi_v^{\geq m} ={} &\exists x_{1,1} \cdots \exists x_{1,n} \cdots \exists x_{m,1} \cdots \exists x_{m,n} \\
      & \bigwedge_{i=1}^{m} \Big(\bigwedge_{j=1}^{n-1} \N(x_{i,j},x_{i,j+1}) 
      \wedge \bigwedge_{j=2}^{n-1}\ (x_{i,j-1} \neq x_{i,j+1})
      \wedge \bigwedge_{j=1}^{n} P_{a_j}(x_{i,j}) \Big) \\
      & {}\wedge \bigwedge_{1 \leq i < j \leq m} 
      \neg((x_{i,1} = x_{j,1} \wedge x_{i,n} = x_{j,n}) \vee 
      (x_{i,1} = x_{j,n} \wedge x_{i,n} = x_{j,1})) \,. 
    \end{split}
  \end{equation*}
  Similarly, we can write a formula $\psi_v\in\fo(\N)$ that says that a word 
  belongs to $\{v,v^{r}\}$. Finally, given two words of same length $u,v\in 
  A^{n}$, we can write a formula $\chi_{u,v}\in\fo(\N)$ that says that $u,v$
  occur at two different end points of a word $w$, i.e.\ that 
  $\{x,y^{r}\}=\{u,v\}$ where $x,y$ are the prefix and suffix of $w$ of length 
  $n$.

  ($\Rightarrow$) Hanf's theorem \cite{DBLP:books/daglib/0082516} states that
  two structures $A$ and $B$ are $m$-equivalent (i.e.\ indistinguishable by any
  \fo formula of quantifier rank at most $m$), for some $m \in \mathbb{N}$ if
  for each $3^m$ ball type $S$, both $A$ and $B$ have the same number of $3^m$
  balls of type $S$ upto a threshold $m\times e$, where $e \in \mathbb{N}$.
  Applying Hanf's theorem to undirected path graphs, we obtain that given an
  $\fo(\N)$ formula $\Phi$, there exist $k,t>0$ such that the fact that a word
  $w$ satisfies $\Phi$ only depends on its $\approxr t k$-class.  The set of all
  such words is therefore an \lrtt language.
\end{proof}

\section{Semigroups with Involution}
\label{Section:semigroup}
In this section we address the question of definability of a language --- ``is
the given reversible regular language definable by a formula in the logic?"  ---
in the previously defined logics.  We show that in the case of $\fo(\bet)$ the
existing theorems provide an algorithm for the problem, while for $\fo(\N)$ the
answer is not yet known.

First we recall the notion of recognisability by a finite semigroup. A finite
semigroup $(S,\cdot)$ is a finite set $S$ with an associative binary operation
${\cdot}\colon S \times S \rightarrow S$. If the semigroup operation has an identity,
then it is necessarily unique and is denoted by $1$.  In this case $S$ is called
a monoid.  A semigroup morphism from $(S, \cdot)$ to $(T, +)$ is a map $h\colon S
\rightarrow T$ that preserves the semigroup operation, i.e.\ $h(a \cdot b)= h(a)
+ h(b)$ for $a,b$ in $S$.  Further if $S$ and $T$ are monoids the map is a
monoid morphism if $h$ maps the identity of $S$ to the identity of $T$.

The set $A^*$ ($\resp$ $A^+$) under concatenation forms a free monoid ($\resp$
free semigroup).  A language $L \subseteq A^*$ is {\em recognised} by a
semigroup (or monoid) $(S, \cdot)$, if there is a morphism $h\colon A^*
\rightarrow (S, \cdot)$ and a set $P \subseteq S$, such that $L = h^{-1}(P)$.

Given a language $L$, the {\em syntactic congruence} of $L$, denoted as $\sim_L$
is the congruence on $A^*$,
$$
x \sim_L y ~~\text{ if }~~ uxv \in L \Leftrightarrow uyv \in L \text{ for all $u,v \in A^*$}.
$$
The quotient $A^*\!/\!\sim_L$, ($\resp$ $A^+\!/\!\sim_L$) denoted as $M(L)$, is
called the {\em syntactic monoid} ($\resp$ {\em syntactic semigroup}).  It
recognises $L$ and is the unique minimal object with this property: any monoid
$S$ recognising $L$ has a surjective morphism from a submonoid of $S$ to $M(L)$
\cite{StraubingBook}.

In the particular case of reversible languages the syntactic monoid 
described above admits further properties.  The observation is that the reverse
operation can be extended to congruence classes of the syntactic congruence by
letting $[x]^r = [x^r]$ for each word $x$ and it is well defined since if $x
\sim_L y $ then $x^r \sim_L y^r$ as can be easily verified.  Moreover this operation is
an involution, i.e.\ $\([x]^r\)^r = \([x^r]\)^r = \left [\(x^r\)^r\right ] =
[x]$, and an anti-isomorphism on the congruence classes, i.e.\ $\([x]\cdot
[y]\)^r = \([x \cdot y]\)^r= [ (x \cdot y)^r ] = \left[ y^r \cdot x^r\right ] =
\left [ y^r \right ] \cdot \left [ x^r \right ] = \left [ y \right ]^r \cdot
\left [ x \right ]^r$.  Therefore one can enrich the notion of semigroups for
recognisability in the case of reversible languages as below.

A {\em semigroup with involution} (also called a $\star$-semigroup) $(S, \cdot,
\star)$ is a semigroup $(S, \cdot)$ extended with an operation ${\star} \colon S
\rightarrow S$ (called the involution) such that
\begin{enumerate}
  \item the operation $\star$ is an involution on $S$, i.e.\ $\(a^\star\)^\star
  = a$ for all elements $a$ of $S$, 

  \item the operation $\star$ is an anti-automorphism on $S$ (isomorphism
  between $S$ and opposite of $S$), i.e.\ $\(a \cdot b\)^\star = b^\star \cdot
  a^\star$ for any $a,b$ in $S$.
\end{enumerate}
It is a $\star$-monoid if $S$ is a monoid.  It is easy to see that in the case
of $\star$-monoids, necessarily $1^\star =1$.  Clearly the free monoid $A^*$
with the reverse operation $r$ as the involution is a $\star$-monoid, since
$(w^r)^r = w$ and $(v \cdot w)^r = w^r \cdot v^r$.  When there is no
ambiguity, we just write $A^*$ to refer to the $\star$-monoid $(A^*, \cdot, r)$.

A map $h\colon S \rightarrow T$ between two $\star$-semigroups $(S, \cdot,
\star)$ and $(T, +, \dagger)$ is a morphism if it is a morphism between the
semigroups $(S,\cdot)$ and $(T,+)$ that preserves the involution, i.e.\
$h(a^\star) = h(a)^{\dagger}$.

A language $L \subseteq A^*$ is said to be recognised by a $\star$-semigroup
$(S, \cdot ,\star)$, if there is a morphism $h\colon (A^*,\cdot, r) \rightarrow
(S,\cdot,\star)$ and a set $P \subseteq S$, such that $P^\star = P$ and $L =
h^{-1}(P)$.  The following proposition summarises the discussion so far.

\begin{proposition}
The following are equivalent for a language $L$.
\begin{enumerate}
\item $L$ is a reversible regular language,
\item $L$ is recognised by a finite $\star$-monoid,
\item $M(L)$ with the reverse operation is a finite $\star$-monoid with 
$P=P^{\star}$ where $P=\{[u]\mid u\in L\}$,
i.e.\ $(M(L),\cdot,r)$ recognises $L$ as a $\star$-monoid.
\end{enumerate}
\end{proposition}

A semigroup (or monoid) is aperiodic if there is some $n \in \mathbb{N}$ such
that $a^n=a^{n+1}$ for each element $a$ of the semigroup.
Sch\"utzenberger-McNaughton-Papert theorem states that a language $L$ is
definable in $\fo(<)$ if and only if the syntactic monoid is aperiodic.  This
theorem in conjunction with Proposition \ref{Prop:fobet} gives that,
\begin{proposition}
  A reversible language $L$ is definable in $\fo(\bet)$ if and only if 
  $M(L)$ is aperiodic.
\end{proposition}

The above theorem hence yields an algorithm for definability of a language in
$\fo(\bet)$, i.e.\ check if the language is reversible, if so compute the
syntactic monoid (which is also a monoid with an involution) and test for
aperiodicity.

Next we look at the logic $\fo(\N)$.  The characterisation theorem for $\fo(+1)$
due to Brzozowski and Simon \cite{BrozozoskiSimon}, and Beauquier and Pin
\cite{BeauquierPin}, is stated below.  Recall that an element of a semigroup $e$
is an {\em idempotent} if $e\cdot e = e$.

\begin{theorem}[Brzozowski-Simon, Beauquier-Pin] 
The following are equivalent.
\begin{enumerate}
  \item $L$ is locally threshold testable.

  \item $L $ is definable in $\fo(+1)$.

  \item The syntactic semigroup 
  of $L$ is finite, aperiodic and satisfies the identity
  $e\,{x}\,f\,y\,e\,{z}\,f\ =\
  e\,{z}\,f\,y\,e\,{x}\,f$
  for all $e,f,x, y, z \in M(L)$ with $e,f$ idempotents.
\end{enumerate}
\end{theorem}

Because of Proposition \ref{Prop:fon} we need to add more identities to
characterise the logic $\fo(\N)$ in terms of $\star$-semigroups.

\begin{theorem}
  The syntactic $\star$-semigroup 
  of an $\fo(\N)$-definable language
  satisfies the identity
  $$e{x} e^\star = e {x^\star} e^\star,$$
  where $e$ is an idempotent, and $x$ is any element of the semigroup.
\end{theorem}

\begin{proof}
  Assume we are given an $\fo(\N)$-language $L$, with its syntactic
  $\star$-semigroup $M=\(A^+\!/\!\sim_L, \cdot, \star\)$, and $h\colon A^+
  \rightarrow M$ the canonical morphism recognising $L$.  Let $e$ be an
  idempotent of $M$, and let $x$ be an element of $M$.  Pick nonempty words $u$
  and $s$ such that $h(u)=e$ and $h(s)=x$.

  By definition of the involution, $h(u^r)=e^\star$ and $h(s^r)=x^\star$.  We
  are going to show that 
  $usu^r \sim_L us^ru^r$ and hence they will correspond to the same
  element in the syntactic $\star$-semigroup, proving that $ex e^\star = e
  x^\star e^\star$.

  Since $L$ is $\fo(\N)$ definable, we know by Theorem \ref{Thm:fon} that $L$ is
  a union of $\approxr t k$ equivalence classes for some $k,t>0$.
  Consider the words $w=(u^k)s(u^k)^r$ and $w^r=(u^k)s^r(u^k)^r$, obtained by
  pumping the words corresponding to $e$ and $e^\star$.  Since
  $e, e^\star$ are idempotents, it is clear that $h(w)=h(usu^r)=exe^{\star}$ and
  $h(w^r)=h(us^ru^r)=ex^{\star}e^{\star}$. 
  
  For all contexts $\alpha,\beta \in A^*$, we show below that $\alpha w \beta
  \approxr t k \alpha w^r \beta$, which implies $\alpha w \beta \in L$ iff $\alpha
  w^r \beta \in L$ since $L$ is a union of $\approxr t k$ classes.  It follows
  that $w\sim_L w^r$ and therefore $h(w)=h(w^r)$, which will conclude the proof.
  
  Fix some contexts $\alpha,\beta\in A^{*}$.  Since $u\neq\varepsilon$, the
  words $\alpha w \beta$ and $\alpha w^r \beta$ have the same prefix of length
  $k-1$ and the same suffix of length $k-1$.  Now, consider $v\in A^{k}$.  If an
  occurrence of $v$ (resp.\ $v^{r}$) in $\alpha w \beta$ overlaps with $\alpha$
  or $\beta$ then we have the very same occurrence in $\alpha w^{r} \beta$.
  Using $w \approxr t k w^{r}$, we deduce that $\sharp^{r}(\alpha w
  \beta,v)=^{t} \sharp^{r}(\alpha w^r \beta,v)$.  Therefore, $\alpha w \beta
  \approxr t k \alpha w^r \beta$.
\end{proof}
The converse direction is open.  The similar direction in the case of $\fo(+1)$
goes via categories \cite{Tilson} and uses the Delay theorem of Straubing
\cite{StraubingVD,StraubingBook}.

\section{Conclusion} \label{Section:conclusion}

The logics $\mso(\bet), \mso(\N)$ and $\fo(\bet)$ behave analogously to the
classical counterparts $\mso(<), \mso(+1)$ and $\fo(<)$.  But the logic
$\fo(\N)$ gives rise to a new class of languages, locally-reversible threshold
testable languages.  The quest for characterising the new class takes us to the
formalism of involution semigroups.  The full characterisation of the new class
is the main question we leave open.  Another line of investigation is to study
the equationally-defined classes that arise naturally from automata theory.


\end{document}